\documentclass[runningheads]{llncs}

\usepackage[utf8]{inputenc}
\usepackage{marvosym}
\usepackage{float}
\usepackage{tikz}
\usetikzlibrary{calc,decorations.pathmorphing,shapes,arrows}
\usetikzlibrary{shapes}
\usetikzlibrary{automata}
\usetikzlibrary{calc}
\usepackage{amssymb}
\usepackage{mathpartir}
\usepackage{comment}
\usepackage{multicol}
\usepackage{scalerel}
\usepackage{bussproofs}
\usepackage{lscape}
\usepackage{stmaryrd}
\usepackage[title]{appendix}
\usepackage{ upgreek }
\usepackage{listings}
\usepackage{stackengine}
\usepackage[tworuled, vlined, noend]{algorithm2e}
\usepackage{macro}
\usepackage{capt-of}
\usepackage{enumerate}

\usepackage{mathtools}
\usepackage{cmll}
\usepackage{mathrsfs}
\usepackage{graphicx}
\usepackage{thm-restate}
\usepackage{hyperref}
\renewcommand{\orcidID}[1]{\smash{\href{http://orcid.org/#1}{\protect\raisebox{-1.25pt}{\protect\includegraphics{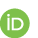}}}}}

\begin{document}

\title{Guessing the buffer bound for k-synchronizability}

\author{Cinzia {Di Giusto}
       \inst{1}
       \orcidID{0000-0003-1563-6581}
\and   Laetitia Laversa
       \inst{1}
       \orcidID{0000-0003-3775-6496}
\and   Etienne Lozes
       \inst{1}
       \orcidID{0000-0001-8505-585X}
}

\institute{Université Côte d’Azur, CNRS, I3S, Sophia Antipolis, France}
\authorrunning{C. Di Giusto et al.}



%
\maketitle             
\begin{abstract}
A communicating system is \kSable{k} if all of the message sequence charts representing the executions can be divided into slices of  $k$ sends followed by $k$ receptions.
It was previously shown that, for a fixed given $k$, one could decide whether a  communicating system is \kSable{k}. This result
is interesting because the reachability pro\-blem can be solved for \kSable{k} systems. However, the decision procedure assumes that the bound
$k$ is fixed.
In this paper we improve this result and show that it is possible to decide if such a bound $k$ exists.
\keywords{communicating automata \and MSC \and synchronizability}
\end{abstract}

\section{Introduction}
Communicating finite state machines \cite{DBLP:journals/jacm/BrandZ83} model distributed systems where participants exchange messages via FIFO buffers.
Due to the unboundedness of the buffers, the model is Turing powerful as soon as there are two participants and two queues. In order to recover decidability, several works introduced restrictions on the model, for instance: lossiness of the channels~\cite{AbdullaJ96}, specific topologies, or bounded
context switching~\cite{DBLP:conf/tacas/TorreMP08}. Another line of research focused on analyzing the system under the assumption that the semantics is synchronous
\cite{BasuB16,DBLP:conf/icalp/FinkelL17,Lipton75,ElradF82,ChouG88,Chaouch-SaadCM09,KraglQH18,GleissenthallKB19}
or that buffers are bounded. This assumption is not as res\-trictive as it may seem at first, because several systems enjoy the property that their execution, although not necessarily bounded, can be simulated by a causally equivalent bounded execution.
Existentially $k$-bounded communica\-ting systems~\cite{DBLP:journals/fuin/GenestKM07} are precisely the systems whose message sequence charts
can be generated by $k$-bounded executions. In particular, the reachability problem is decidable for existentially $k$-bounded communicating systems.
A limitation of this framework is that the bound $k$ on the buffer size must be fixed. A natural question is whether the existence of such a bound can be decided.
Genest, Kuske and Muscholl answered this question negatively~\cite{DBLP:journals/fuin/GenestKM07}.
Bouajjani et al.~\cite{DBLP:conf/cav/BouajjaniEJQ18}\footnote{The results in \cite{DBLP:conf/cav/BouajjaniEJQ18} have then been refined in \cite{DBLP:conf/fossacs/GiustoLL20}}
introduced a variant of existentially $k$-bounded communicating systems they called \kSable{k} systems.
A system is \kSable{k} if each of its execution is causally equivalent to a sequence of communication rounds composed of
at most $k$ sends followed by at most $k$ receptions. In particular, each execution of a \kSable{k} system is causally equivalent to a $k$-bounded
execution (provided all messages are eventually received). Like for existentially bounded systems, the reachability problem becomes decidable
for \kSable{k} systems, and the membership problem - whether a given system is \kSable{k} for a fixed given $k$ is decidable as well.
Bouajjani et al. conjectured that  the existence of a bound $k$ on the size of the
communication rounds was undecidable.

Instead, in this paper, we show  that this problem is decidable.
This result contrasts with the negative result about the same question for existentially bounded communicating systems.
There is an important difference between existentially bounded
 and \kSable{k} ones that explains this situation.
Existentially bounded  systems deal with peer-to-peer communications, with one buffer per pair of machines,
whereas \kSable{k} systems deal with mailbox communications where one buffer per machine merges all incoming
messages.

The paper is organized as follows:
in the next section, we introduce
preli\-mi\-nary definitions on communicating automata
and \kSable{k} systems.
In Section~\ref{sec:lpfe} we explain the general strategy for computing
the bound $k$, which is to compute the automata of two regular
languages: the language of reachable exchanges, and the language of
prime exchanges. In Section~\ref{sec:rfe} we focus on rea\-chable exchanges,
and in Section~\ref{sec:prime} on prime exchanges.
Section~\ref{sec:computation-k0} lastly computes the bound $k$.
Finally Section \ref{sec:conc} concludes with some final remarks.
An appendix with additional material and proofs is added for the reviewer convenience.

\section{Preliminaries}\label{sec:prelim}

Let $\paylodSet$ be a finite set of messages and $\procSet$ a finite set of processes exchanging messages.
A send action, denoted $\send{p}{q}{\amessage}$, designates the sending of message $\amessage$ from process $p$ to process $q$, storing it in the queue of $q$. Similarly, a receive action $\rec{p}{q}{\amessage}$ expresses that process $q$  pops message $\amessage$ from its queue of incoming messages.  We write $a$ to denote a send or receive action.
Let
$\sendSet~=~\{ \send{p}{q}{\amessage} \mid  p,q \in \procSet, \amessage \in \paylodSet \}$ be
the set of send actions and $\receiveSet~=~\{ \rec{p}{q}{\amessage}\mid  q \in \procSet, \amessage \in \paylodSet\}$ the set of receive actions. $\sendSet_p$ and $\receiveSet_p$ stand for the set of sends and receives  of process $p$ respectively.

A \emph{system} is a tuple $ \system = \left( (L_p, \delta_p, l^0_p) \mid p \in \procSet \right) $ where, for each process $p$, $L_p$ is a finite set of local
control states, $\delta_p \subseteq (L_p \times ( \sendSet_p \cup \receiveSet_p ) \times L_p)$ is the transition relation
and $l^0_p$ is the initial state.
In the rest of the paper, when talking about a system $\system$, we may  also identify it with the global automaton obtained as  the product of the process automata  and denoted  $(L_\system, \delta_\system, \globalstate{l_0})$ where
$L_\system = \Pi_{p \in \procSet} L_p$ is the set of global control states,
    $\globalstate{l_0} = (l^0_p)_{p\in \procSet}$ is the initial global control
state and
$((l_1,\cdots, l_q, \cdots, l_n), a, (l_1,\cdots, l'_q,
\cdots, l_n)) \in \delta_\system$ iff $ (l_q, a, l'_q) \in \delta_q$ for $ q \in \procSet$.
We write $\globalstate{l}$ in bold
to denote the tuple of control states $(l_p)_{p\in\procSet}$, and we sometimes
write $l_q \xrightarrow{a}_q l'_q$ (resp. $\globalstate{l}\tr{a}\globalstate{l}')$ for $(l_q,a,l_q')\in\delta_q$ (resp. $(\globalstate{l},a,\globalstate{l}')\in\delta_{\system}$). We write $\Tr{a_1 \cdots a_n}$ for $\tr{a_1} \cdots \tr{a_n}$.

A \emph{configuration} is a pair $(\globalstate{l}, \B)$
where $\globalstate{l}=(l_p)_{p\in\procSet} \in  L_\system$ 
  is a global control state of $\system$,
and $\B= (b_p)_{ p \in \procSet} \in (\paylodSet^*)^\procSet$ is a vector of
buffers, each $b_p$ being a word over $\paylodSet$.
$\B_0$ stands for the vector
of empty  buffers.
The mailbox semantics of a system is defined by the two rules below.

\noindent \begin{minipage}{.5\columnwidth}
\begin{center}
\small{[SEND]}

\AxiomC{{\small $
\globalstate{l} \tr{\send{p}{q}{\amessage}} \globalstate{l'}
    \quad
    b_{q}' = b_{q} \cdot \amessage$}}
\UnaryInfC{{\small
$(\globalstate{l}, \B)
               \trB{\send{p}{q}{\amessage}}
	       (\globalstate{l}',\B\sub{b_{q}'}{b_{q}})
               $
}}
\DisplayProof
\end{center}
\end{minipage}
\begin{minipage}{.5\columnwidth}
\begin{center}
\small{[RECEIVE]}

\AxiomC{{\small$ \globalstate{l}
               \tr{\rec{p}{q}{\amessage}}
               \globalstate{l}'
    \quad
    b_{q} = \amessage \cdot b_{q}'$}}
\UnaryInfC{{\small$(\globalstate{l}, \B)
               \trB{\rec{p}{q}{\amessage}}
	       (\globalstate{l}',\B\sub{b_{q}'}{b_{q}})
               $}}
\DisplayProof
\end{center}
\end{minipage}

\vspace*{0.1cm}
\noindent In this paper, we focus on \textbf{mailbox} semantics.
An  execution $e=a_1\cdots a_n$ is a sequence of actions in
$\sendSet \cup \receiveSet$ such that $(\globalstate{l_0}, \B_0) \trB{a_1} \cdots \trB{a_n} (\globalstate{l}, \B)$ for some $\globalstate{l}$ and $\B$.
As usual, $\TrB{e}$
stands for $ \trB{a_1} \cdots \trB{a_n}$.
We write $\asEx(\system)$ to denote the set of executions of
a system $\system$.
Executions impose a total order over the actions. To stress the causal dependencies between messages we use message sequence charts (MSCs) that only impose an order between matched pairs of actions and between the actions of a same
process.

\begin{definition}[Message Sequence Chart]
A message sequence chart $\msc$  is a tuple $(Ev,\lambda, \prec_{po}, \prec_{src})$ such that
\begin{enumerate}
\item 
 $Ev$ is a finite set of events partially ordered under $(\prec_{po}\cup\prec_{src})^{*}$,
\item 
 $\lambda:Ev\to S\cup R$ tags each event with an action,
\item 
 for each process $p$, $\prec_{po}$ induces a total order on the events of $p$, i.e. on $\lambda^{-1}(S_p\cup R_p)$,
\item 
 $(Ev,\prec_{src})$ is the graph of a bijection between a subset of $\lambda^{-1}(S)$ and the whole of $\lambda^{-1}(R)$
\item 
 for all $s\prec_{src} r$, there are $p,q,\amessage$ such that
$\lambda(s)=\send{p}{q}{\amessage}$ and $\lambda(r)=\receive{q}{\amessage}$.
\end{enumerate}
\end{definition}

\begin{definition}[Concatenation of MSCs]
Let $\msc_1=(Ev_1,\lambda_1,\prec_{po}^1,\prec_{scr}^1)$ and
$\msc_2=(Ev_2,\lambda_2,\prec_{po}^2,\prec_{scr}^2)$ be two MSCs. Their
concatenation $\msc_1\cdot\msc_2$ is the MSC $\msc=(Ev,\lambda,\prec_{po},\prec_{src})$ such that:
\begin{itemize}
\item $Ev=Ev_1\cup Ev_2$
\item $\lambda=\lambda_1\cup\lambda_2$
\item $\prec_{po}=\prec_{po}^1\cup\prec_{po}^2\cup\bigcup_{p\in\procSet}\{(e_1,e_2)\mid e_1\in \lambda_1^{-1}(S_p\cup R_p),e_2\in \lambda_2^{-1}(S_p\cup R_p)\}$
\item $\prec_{src}=\prec_{src}^1\cup\prec_{src}^2$.
\end{itemize}
\end{definition}
In a sequence of actions $e=a_1\cdots a_n$,
a send action $a_i= \send{p}{q}{\amessage}$ is  \emph{matched}
by a reception $a_j=\rec{p'}{q'}{\amessage'}$  (denoted by $a_i \matches a_j$)
if $i< j$, $p=p'$, $q=q'$, $\amessage=\amessage'$, and there is $\ell\geq 1$ such that
$a_i$ and $a_j$ are the $\ell$th actions of $e$ with these properties respectively.
A send action $a_i$ is \emph{unmatched} if there is no
matching reception in $e$.

The MSC associated with the execution $e=a_1\cdots a_n$ is $(Ev,\lambda,\prec_{po},\prec{src})$
where $Ev=\{1,\cdots,n\}$, $\lambda(i)=a_i$, $i\prec_{po} j$ iff $i<j$ and $\{a_i,a_j\}\subseteq S_p\cup R_p$ for some $p$, and $i\prec_{src} j$ if $a_i\matches a_j$.
%

When $\amessage$ is either an unmatched
$\send{p}{q}{\amessage}$ or a pair of matched actions
$\{\send{p}{q}{\amessage},\rec{p}{q}{\amessage}\}$, we write
$\sender{\amessage}$ for $p$ and $\receiver{\amessage}$ for $q$.
Note that $\procofactionv{R}{\amessage}$ is defined even if $\amessage$ is unmatched.
An MSC is depicted with vertical
timelines (one for each process) where time goes from top to bottom. Points on the lines represent events of this
process. We draw an arc between two matched events and a dashed arc to depict an
unmatched send. The concatenation $\msc_1\cdot \msc_2$ of two MSCs is the union of the two MSCs where,
for each $p$, all $p$-events of $\msc_1$ are considered $\prec_{po}$ smaller than all $p$-events
of $\msc_2$. 
  We write $msc(e)$ for the MSC associated with the execution $e$, 
  and we say that a
  sequence of actions $e$ is a linearization of a given MSC if it is the sequence of actions induced by a total
  order extending $(\prec_{po}\cup\prec_{src})^*$.
We write $\asTr(\system)$ for the set $\{msc(e)\mid e\in\asEx(\system)\}$.
We write $\globalstate{l}\trMSC{\mu}\globalstate{l}'$ to denote
that $\globalstate{l}\Tr{e}\globalstate{l'}$
for any linearization $e$ of $\mu$.
Finally, we recall from \cite{DBLP:conf/fossacs/GiustoLL20} the definition of causal delivery that allows to consider only
 MSCs that correspond to executions in the mailbox semantics.
\begin{definition}[Causal delivery]\label{def:causal-delivery}
Let $\msc=(Ev , \lambda, \prec_{po},\prec_{src})$ be an MSC.
We say that $\msc$ satisfies causal delivery if
it admits a linearization with the total order $<$ such that
for any two  events $s_1, s_2\in Ev$, if $s_1<s_2$,
$\lambda(s_1)~=~\send{p}{q}{\amessage}$ and
$\lambda(s_2) = \send{p'}{q}{\amessage'}$ for a same
destination process $q$,  then
either $s_2$ is unmatched, or there are $r_1,r_2$
such that $s_1\prec_{src}r_1$, $s_2\prec_{src}r_2$, and
$r_1< r_2$.
\end{definition}

  A \kE{k} (with $k\geq1$) is an MSC that admits a linearization $e\in S^{\leq k}R^{\leq k}$ starting with
  at most $k$ sends and followed by at most $k$ receives. An MSC is \emph{\kSous{k}} if it can be chopped into a sequence of \emph{\kE{k}s}.

%
%
\begin{definition}[\kSous{k}]
  An MSC $\msc$ is \kSous{k} if $\msc=\msc_1\cdot \msc_2\cdots \msc_n$ where, for all $i\in\interval{1}{n}$,
$\msc_i$ is a $\kE{k}$.
\end{definition}

\vspace*{0.25cm}
\noindent
\begin{minipage}[c]{9.5cm}


For instance, the MSC $\mu_1$ depicted on Fig.~\ref{fig:msc_2-synchro} is \kSous{2}, as it can be split in two \kE{2}s.

\hspace*{.5cm}An execution $e$ is \kSable{k} if $msc(e)$ is \kSous{k}. A system $\system$ is \emph{\kSable{k}} if
all its executions are \kSable{k}.


\end{minipage}
\begin{minipage}[c]{3cm}
  \begin{center}
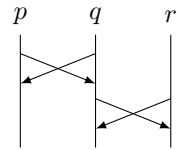

  \begin{tikzpicture}
  \coordinate (pa) at (0,-0.25) ;
  \coordinate (pb) at (0,-1.75) ;
  \coordinate (qa) at (1,-0.25) ;
  \coordinate (qb) at (1,-1.75) ;
  \coordinate (ra) at (2, -0.25);
  \coordinate (rb) at (2, -1.75);
  \draw (0,0) node{$p$} ;
  \draw (1,0) node{$q$} ;
  \draw (2,0) node{$r$} ;
  \draw (pa) -- (pb) ;
  \draw (qa) -- (qb) ;
  \draw (ra) -- (rb) ;
  \draw[>=latex,->] (0,-0.5) -- (1,-0.9);
  \draw[>=latex,->] (1,-0.5) -- (0,-0.9);
  \draw[>=latex,->] (1,-1.1) -- (2,-1.5);
  \draw[>=latex,->] (2,-1.1) -- (1,-1.5);
  \end{tikzpicture}
  \vspace*{-0.3cm}\captionof{figure}{MSC $\mu_1$}
  \label{fig:msc_2-synchro}
\end{center}
\end{minipage}

\begin{theorem}[\cite{DBLP:conf/cav/BouajjaniEJQ18,DBLP:conf/fossacs/GiustoLL20}]
\label{thm:ksync-for-fixed-k}
  It is decidable whether a system $\system$ is \kSable{k} for a  given $k$.
Moreover, it is decidable to know whether a control state is reachable
under the assumption that $\system$ is \kSable{k}.
\end{theorem}

This result is interesting but somehow incomplete
as it assumes that a fixed value of the parameter $k$
has been found.
We aim at answering this
limitation by
computing the \emph{synchronizability degree}
of a given system.

\begin{definition}[Synchronizability degree]
  The synchronizability degree $\sd{\system}$ of
  a system $\system$ is the smallest $k$ such
  that $\system$ is \kSable{k}. In particular,
  $\sd{\system}=\infty$ if $\system$ is not
  \kSable{k} for any $k$.
\end{definition}

\section{Largest prime reachable exchange\label{sec:lpfe}\label{sec:lpfe}}

In this section, we relate the synchronizability degree of a system to
the size of a ``maximal, prime, reachable exchange".
We start with defining these notions.

An \emph{exchange} is a \kE{k} for some arbitrary $k$, and
we call $k$ the size of the exchange.
An exchange $\msc$ is \emph{reachable} if there exist exchanges
$\msc_1,\cdots,\msc_n$ for some $n\geq 0$
and such that $\msc_1\cdots\msc_n\cdot\msc\in\asTr(\system)$.
An exchange $\msc$ is \emph{prime} if there does not exist a
decomposition $\msc=\msc_1\cdot \msc_2$
in two non-empty exchanges. For instance, the 2-exchange
(depicted by the MSC $\mu_2$,
Fig.~\ref{fig:msc_exchange_not_prime}) with linearization:

\begin{minipage}[c]{9cm}
  $$
  \send{p}{q}{\amessage_1}\cdot\send{r}{q}{\amessage_2}\cdot\receive{q}{\amessage_1}\cdot\receive{q}{\amessage_2}
  $$
  \hspace*{-0.5cm}is not prime, as it can be factored in two 1-exchanges as follows
  $$
  \send{p}{q}{\amessage_1}\cdot\receive{q}{\amessage_1}~~\cdot~~\send{r}{q}{\amessage_2}\cdot\receive{q}{\amessage_2}.
  $$
\end{minipage}
\begin{minipage}[c]{3cm}
  \begin{center}
  \begin{tikzpicture}
    \coordinate (pa) at (0,-0.25) ;
    \coordinate (pb) at (0,-1.5) ;
    \coordinate (qa) at (1,-0.25) ;
    \coordinate (qb) at (1,-1.5) ;
    \coordinate (ra) at (2, -0.25);
    \coordinate (rb) at (2, -1.5);
    \draw (0,0) node{$p$} ;
    \draw (1,0) node{$q$} ;
    \draw (2,0) node{$r$} ;
    \draw (pa) -- (pb) ;
    \draw (qa) -- (qb) ;
    \draw (ra) -- (rb) ;
    \draw[>=latex,->] (0,-0.75) -- (1,-0.75) node[midway, above]{$\amessage_1$};

    \draw[>=latex,->] (2,-1.25) -- (1,-1.25)  node[midway, above]{$\amessage_2$};

  \end{tikzpicture}
\vspace*{-0.3cm}  \captionof{figure}{MSC $\mu_2$}
  \label{fig:msc_exchange_not_prime}
\end{center}
\end{minipage}

\vspace*{0.2cm}
The size of the biggest prime reachable
exchange is related to the synchronizability
degree $\sd{\system}$ by the following property.

\begin{restatable}{lemma}{lemmalpfe}\label{lemma-lpfe}
Let $k\in\mathbb{N}\cup\{\infty\}$ be the
supremum of the sizes of all prime reachable exchanges.
(1) If $k=\infty$, then $\sd{\system}=\infty$
 (2) if $k<\infty$, then either
$\system$ is \kSable{k} and $\sd{\system}=k$,
or $\system$ is not \kSable{k} and $\sd{\system}=\infty$.
\end{restatable}


\begin{proof}
Let $k\in \mathbb{N}\cup\{\infty\}$ be the supremum of the sizes of all prime reachable
exchanges.

Assume that there exists $K$ such that $\system$ is
\kSable{K}. Let us show that $k\leq K$ and $\system$ is \kSable{k}.\\
\begin{itemize}
    \item $k\leq K$.
Assume by contradiction that $k\geq K+1$. Then there exists exchanges
$\msc_1,\dots,\msc_n,\msc$ such that
$\msc_1\cdots\msc_n\cdot\msc\in\asTr(\system)$ and
$\msc$ is prime of size $K+1$.
Since $\msc$ is prime, it corresponds to a
strongly connected component
of size $K+1$ of the conflict graph of
$\msc_1\cdots\msc_n\cdot\msc$, so
$\msc_1\cdots\msc_n\cdot\msc$ cannot be \kSous{K}:
contradiction.
    \item $\system$ is \kSable{k}.
Let $\msc\in\asTr(\system)$
be fixed an let us show that it can be chopped into a sequence of $k$ exchanges.
Since by hypothesis $\system$ is \kSable{K}, there
are \kE{K}s
$\msc_1,\dots,\msc_n$ such that $\msc=\msc_1\cdots\msc_n$.
Up to decomposing each $\msc_i$ as a product of prime exchanges,
we can assume that all $\msc_i$ are prime. Moreover, they are
all reachable, so their size is bounded by $k$.
As a consequence, $\msc$ can be decomposed
in a sequence of \kE{k}s.
\end{itemize}
\qed
\end{proof}

Since by Theorem~\ref{thm:ksync-for-fixed-k}
it is decidable whether $\system$ is \kSable{k},
it is enough to know $k$ in order to compute $\sd{\system}$.
In order to compute $k$, we have to address two problems:
the number of exchanges is possibly
infinite, and one should examine sequences of arbitrarily many
exchanges.
To solve these issues, we are going to reduce to a problem on
regular languages.
Let $\Sigma=\{!?,!\}\times\paylodSet\times \procSet^2$; for better readability, we write
$!?\amessage^{p\to q}$ (resp. $!\amessage^{p\to q}$) for a $\Sigma$-symbol.
To every $\Sigma$-word $w$ we associate an MSC $msc(w)$ as follows.
Consider the substitutions $\sigma_1:\Sigma\to S$ and $\sigma_2:\Sigma\to R\cup\{\epsilon\}$ such that
$\sigma_1(!?\amessage^{p\to q})=\sigma_1(!\amessage^{p\to q})=\send{p}{q}{\amessage}$,
$\sigma_2(!?\amessage^{p\to q})=\receive{q}{\amessage}$ and $\sigma_2(!\amessage^{p\to q})=\epsilon$.
Then $msc(w)$ is defined as $msc(\sigma_1(w)\sigma_2(w))$. Clearly, it is an exchange (by construction, it admits a
linearization in $S^*R^*$), but more remarkably any reachable exchange can be represented by such a word.

\begin{restatable}{lemma}{lemsurjectiviteencodage}\label{lem:surjectivite-encodage}
For all reachable exchanges $\msc$, 
there exists $w \in \Sigma^*$ s.t. $\msc=msc(w)$.
\end{restatable}


\begin{proof}
Let $\msc$ be a reachable exchange, and let $\msc_1,\dots\msc_n$ be such that
$\msc_1\cdot\msc_2\cdots\msc_n\cdot\msc\in\asTr(\system)$.
There is a linearization of $\msc_1\cdots\msc_n\cdot\msc$
which follows the mailbox semantics. This linearization
induces a linearization $lin(\msc)$ of $\msc$ that also
follows the mailbox semantics. Then $lin(\msc)$ induces
an enumeration
$\send{p_1}{q_1}{\amessage_1},\dots,\send{p_n}{q_n}{\amessage_n}$
of the send events of $\msc$.
Let $w=a_1\dots a_n$ where $a_i$ is either $!?\Message{\amessage_i}{p_i}{q_i}$
if $\send{p_i}{q_i}{\amessage_i}$ is matched in $\msc$, or
$!\Message{\amessage_i}{p_i}{q_i}$ if it is unmatched.
Then, the claim is that $msc(w)=\msc$, or in other words,
$\sigma_1(w)\sigma_2(w)$ is a linearization of $\msc$.
By contradiction, assume it is not. Then there are two events
$e,e'$ such that $e<e'$ in the enumeration $\sigma_1(w)\sigma_2(w)$ but
$(e',e)\in(\prec_{po}\cup\prec_{src})^*$.
\begin{itemize}
\item if $e,e'$ are two send events then $e$ occurs before $e'$ in
$\sigma_1(w)$, i.e. $e$ occurs before $e'$ in $lin(\msc)$, which
is a linearization of $\msc$, and the contradiction with $(e',e)\in(\prec_{po}\cup\prec_{src})^*$.
\item if $e$ is a send event and $e'$ a receive event, then
$(e',e)\in(\prec_{po}\cup\prec_{src})^*$ contradicts the fact that $\msc$
is an exchange.
\item if $e$ is a receive event and $e'$ is a send event, then
$e<e'$ wrt $\sigma_1(w)\sigma_2(w)$ contradicts the definition
of $\sigma_1,\sigma_2$.
\item assume finally that $e$ and $e'$ are receive events.
From  $(e',e)\in (\prec_{po}\cup\prec_{src})^*$,
we deduce that $e'\prec_{po} e$ , because $\msc$ is an exchange.
Let $s,s'$ be the matching
send events of $e,e'$ respectively.
Since $e<e'$ wrt $\sigma_1(w)\sigma_2(w)$, $s<s'$
wrt $\sigma_1(w)\sigma_2(w)$, and therefore $s<s'$ wrt $lin(\msc)$.
But $e'<e$ wrt $lin(\msc)$ because $e'\prec_{po} e$, which
violates the mailbox semantics: contradiction.
\end{itemize}\qed
\end{proof}

The proof follows from  the fact that it is always possible to receive  messages in the same global
order as they have been sent. Such a property would not hold for peer-to-peer communications, as we can see in the following counter-example. 
\vspace*{0.2cm}
%

\noindent
\begin{minipage}[c]{7.5cm}
  Consider MSC $\mu_6$ on the right.
  This MSC does not satisfy causal delivery in a mailbox semantics, because the sending of $\amessage_1$
  happens before the sending of $\amessage_4$, and the reception of $\amessage_4$ happens before the reception
  of $\amessage_1$. For this reason, there is no word $w$ such that $msc(w)$ corresponds to this MSC: such a word
  would give a linearization that would correspond to a valid mailbox execution.
  On the other hand, this MSC satisfies causal delivery in a peer-to-peer semantics. For instance, the following
  linearization is a peer-to-peer execution:
  \vspace*{-0.4cm}
  $$
  !\amessage_3\cdot!\amessage_4\cdot!\amessage_1\cdot!\amessage_2\cdot?\amessage_2\cdot?\amessage_3\cdot?\amessage_4\cdot?\amessage_1
  $$
\end{minipage}
\begin{minipage}[c]{5.5cm}

\begin{center}
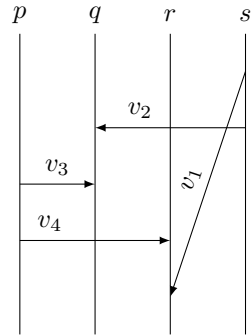

\begin{tikzpicture}

\draw (0,0) node{$p$} ;
\draw (1,0) node{$q$} ;
\draw (2,0) node{$r$} ;
\draw (3,0) node{$s$} ;

  \draw (0,-0.25) -- (0,-4.25) ;
  \draw (1,-0.25)  -- (1,-4.25) ;
  \draw (2,-0.25) -- (2,-4.25) ;
  \draw (3,-0.25) -- (3, -4.25);

\draw[>=latex,->] (3,-0.75) -- node [above,sloped, pos = 0.5] {$\amessage_1$} (2,-3.75);
  \draw[>=latex,->] (3,-1.5) -- node [above,sloped, pos = 0.7] {$\amessage_2$} (1,-1.5);
  \draw[>=latex,->] (0,-2.25) -- node [above,sloped, pos = 0.5] {$\amessage_3$} (1,-2.25);
  \draw[>=latex,->] (0,-3) -- node [above,sloped, pos = 0.2] {$\amessage_4$} (2,-3);
  \end{tikzpicture}
  \captionof{figure}{MSC $\mu_6$}
  \end{center}

  \end{minipage}

We can now define two languages over $\Sigma$:
$$\reachlanguage=\{w\in\Sigma^*\mid msc(w)\mbox{ is reachable}\} \text{ and }
\primelanguage=\{w\in\Sigma^*\mid msc(w) \mbox{ is prime}\}$$
Then the bound $k$ we are looking for is the length of the longest word
in $\reachlanguage\cap \primelanguage$.
It suffices  to show that both $\reachlanguage$
and $\primelanguage$ are effective regular languages
to get an algorithm for computing $k$. This is the content of Sections~\ref{sec:rfe} and \ref{sec:prime}.

\section{Regularity of reachable exchanges\label{sec:rfe}\label{sec:rfe}}

In this section, we aim at defining a finite state automaton that accepts
a word $w\in\Sigma^*$ iff $msc(w)$ is reachable, that is,
iff there exists $\msc_1,\msc_2,\ldots,\msc_n$ such that $
\msc_1\cdot\msc_2\cdots\msc_n\cdot msc(w)\in\asTr(\system).
$
Now, observe that the prefix $\msc_1\cdot\msc_2\cdots\msc_n$
brings the system in a certain global control state that conditions
what can be done by $msc(w)$. Moreover, the presence of unmatched messages in a buffer imposes that none of the subsequent messages sent to the same buffer can be read.

The construction of  the automaton accepting $\reachlanguage$ proceeds
in three separate steps. First, we build an automaton that accepts
the language of all words that code an exchange,
starting  in a certain global control state
$\globalstate{in}$, and ending in another global control state $\globalstate{fin}$, and
possibly not satisfying causal delivery.
Secondly, we consider the set of MSCs that satisfy
causal delivery. We
define automata that recognize the words coding
MSCs starting from a certain ``buffer state'' and ending in
another ``buffer state'', the ``buffer state" characterizing
whether or not the MSC satisfies causal delivery.
Finally, we show that $\reachlanguage$ is a
boolean combination of the languages of some of these automata.


\subsection{Automata of the control states}


We consider triples of global states
$\triple$,  representing the
exchanges
such that
$\globalstatelj$
can be reached only with sends from $\globalstateli$
and $\globalstatelf$
can be reached only with receptions from $\globalstatelj$.
We want to define an automaton $\AutSRauto$ that recognizes
the words coding such exchanges.
Intuitively, $\AutSRauto$ is a product
of on the one hand
the global automaton $\system$ restricted
to send transitions and on the other hand
$\system$ restricted to receive transitions.
For each send action, either the reception is available and a
matched message possible, or there is no corresponding
reception and so we obtain an unmatched message.
\begin{definition}[Automaton of control states] 
	Let $\system$ be a system and $\globalstateli, \globalstatelj,$ $\globalstatelf$  global states.
	$\AutSRauto = (L_{\Aut{SR}}, \delta_{\Aut{SR}}, \globalstate{l^0_{\Aut{SR}}},
	F_{\Aut{SR}})$ is the automaton where:
	\begin{itemize}
		\item $L_{\Aut{SR}} =
		\{ (\globalstate{l}, \globalstate{l'}) \mid~
\globalstate{l}, \globalstate{l'} \in L_{\system} \} $;
		 $\globalstate{l^0_{\Aut{SR}}} = (\globalstateli, \globalstatelj)$; 
		 $ F_{\Aut{SR}} = \{ (\globalstatelj, \globalstatelf) \} $; 
		\item for each $(\globalstate{l_s},\send{p}{q}{\amessage},\globalstate{l'_s}) \in \delta_{\system}$:
		\begin{itemize}
			\item $((\globalstate{l_s}, \globalstate{l}), !\Message{\amessage}{p}{q}, (\globalstate{l'_s}, \globalstate{l}) )\in \delta_{\Aut{SR}}$
			for  $\globalstate{l} \in L_{\system} $;
			\item if $(\globalstate{l_r}, \rec{p}{q}{\amessage}, \globalstate{l'_r}) \in \delta_{\system}$  then $ ((\globalstate{l_s}, \globalstate{l_r}), !?\Message{\amessage}{p}{q},
			(\globalstate{l'_s}, \globalstate{l'_r})) \in \delta_{\Aut{SR}}$
		\end{itemize}
	\end{itemize}
\end{definition}

We denote $\LangSRauto$ the language of a such automaton. This is an example of the construction.
\begin{example}\label{example:aut_SR}
	Let $\system_1$ be the system whose process automata $p,q$ and $r$ are depicted in Fig~\ref{figure:ex_aut_SR}.
	For the triple $\triple$ where $\globalstateli = (0,0,0), \globalstatelj =
	(2,0,1)$ and $\globalstatelf = (2,1,2)$,  
	automaton  $\AutSRauto$ is depicted below the system
	and has for language:
\begin{align*}
	\LangSRauto =~ &
	\Action{!?}{\Message{a}{p}{r}} (
\Action{!}{ \Message{c}{p}{q}} \Action{!?}{\Message{b}{r}{q}} ~+~
\Action{!?}{\Message{b}{r}{q}} \Action{!}{ \Message{c}{p}{q}}) \\
	& + ~ \Action{!?}{\Message{b}{r}{q}} \Action{!?}{\Message{a}{p}{r}}
	\Action{!}{ \Message{c}{p}{q}}
\end{align*}
\vspace*{-1cm}

		\begin{figure}[h]
			\begin{center}
		\begin{tikzpicture}[>=stealth,node distance=2.5cm,shorten >=1pt,
    every state/.style={text=black, scale =0.6}, semithick, scale= 0.7]
  \begin{scope}[->]
      \node[state,initial,initial text={}] (q0)  {$0$};
      \node[state, right of=q0] (q1)  {$1$};
	\node[state, right of=q1] (q2) {$2$};

    \path (q0) edge node [above] {$\Action{!}{\Message{a}{p}{r}}$} (q1);
	\path (q1) edge node [above]    {$\Action{!}{\Message{c}{p}{q}}$}(q2);
    \node[rectangle, thick, draw] at (-0.7,0.6) {$p$};

  \end{scope}

  \begin{scope}[->, xshift=6cm]
      \node[state,initial,initial text={}] (q0)  {$0$};
      \node[state, right of=q0] (q1)  {$1$};

  	\path (q0) edge  node [above]   {$\Action{?}{\Message{b}{r}{q}}$}(q1);
    \node[rectangle, thick, draw] at (-0.7,0.6) {$q$};
  \end{scope}

  \begin{scope}[->, xshift= 10cm]
      \node[state,initial,initial text={}] (q0)  {$0$};
      \node[state,  right of=q0] (q1)  {$1$};
      \node[state,  right of=q1] (q2) {$2$};

    \path (q0) edge node [above] {$\Action{!}{\Message{b}{r}{q}}$} (q1);
    \path (q1) edge node [above] {$\Action{?}{\Message{a}{p}{r}}$} (q2);
    \node[rectangle, thick, draw] at (-0.7,0.6) {$r$};
  \end{scope}

\end{tikzpicture}
		\begin{tikzpicture}[>=stealth,node distance=2.5cm,shorten >=1pt,
    every state/.style={text=black, scale =0.6}, semithick, scale= 0.7]
  \begin{scope}[->, yshift = -6cm, node distance=3.8cm, xshift = 6cm]
      \node[rectangle, thick, draw] at (-5.5,4.5) {$\AutSRauto$};
   	\node[state,
    text width=1.2cm, align = center] (q00)
    {$((0,0,0),$ $(2,0,1))$ 
    };
  	\node[state, right of=q00, text width=1.2cm, align = center] (q10)
    {$((1,0,0),$ $(2,0,1)) $ 
    };
  	\node[state, right of=q10, text width=1.2cm, align = center] (q20)
    {$((2,0,0),$ $(2,0,1)) $ 
    };
  	\node[state, above of=q20, node distance = 3cm, text width=1.2cm,  align = center] (q30)
    {$((2,0,1),$ $(2,0,1)) $
    };
  	\node[state, above of=q10, node distance = 3cm, text width=1.2cm, align = center] (q40)
    {$((1,0,1),$ $(2,0,1)) $ 
    };
  	\node[state, below of=q10, node distance = 3cm, text width=1.2cm, align = center] (q43)
    {$((1,0,1),$ $(2,1,1)) $ 
    };
    \node[state, above of=q00, node distance = 3cm, text width=1.2cm, align = center] (q50)
    {$((0,0,1),$ $(2,0,1)) $ 
    };
   	\node[state, left of=q00, text width=1.2cm, align = center] (q11)
    {$((1,0,0),$ $(2,0,2)) $ 
    };
    \node[state, left of=q11, text width=1.2cm, align = center] (q21)
    {$((2,0,0),$ $(2,0,2)) $ 
    };
  	\node[state, below of=q20, node distance = 3cm, text width=1.2cm, align = center] (q33)
    {$((2,0,1),$ $(2,1,1)) $ 
    };
   	\node[state, above of=q21, node distance = 3cm, text width=1.2cm, align = center] (q31)
    {$((2,0,1),$ $(2,0,2)) $ 
    };
  	\node[state, right of=q31, text width=1.2cm, align = center] (q41)
    {$((1,0,1), $ $ (2,0,2)) $ 
    };
  	\node[state, below of=q00, node distance = 3cm, text width=1.2cm, align = center] (q53)
    {$((0,0,1), $ $(2,1,1))$ 
    };
  	\node[state, below of=q11, node distance = 3cm, text width=1.2cm, align = center] (q42)
    {$((1,0,1), $ $(2,1,2)) $ 
    };
  	\node[state, below of=q21,node distance = 3cm, text width=1.2cm, align = center, accepting] (q32)
    {$((2,0,1), $ $(2,1,2)) $ 
    };

   \path (-0.85,-0.85) edge node{} (-0.55,-0.55);

  	\path (q00) edge node[above] {$\Action{!}{\Message{a}{p}{r}}$} (q10);
  	\path (q00) edge node[left] {$\Action{!}{\Message{b}{r}{q}}$} (q50);
  	\path (q00) edge node[above] {$\Action{!?}{\Message{a}{p}{r}}$}  (q11);
  	\path (q00) edge node[left] {$\Action{!?}{\Message{b}{r}{q}}$} (q53);

  	\path (q10) edge node[above] {$\Action{!}{\Message{c}{p}{q}}$} (q20);
  	\path (q10) edge node[left] {$\Action{!?}{\Message{b}{r}{q}}$}  (q43);
   	\path (q10) edge node[left] {$\Action{!}{\Message{b}{r}{q}}$} (q40);

  	\path (q11) edge node[above] {$\Action{!}{\Message{c}{p}{q}}$} (q21);
  	\path (q11) edge  node[left] {$\Action{!}{\Message{b}{r}{q}}$} (q41);
  	\path (q11) edge node [left] {$\Action{!?}{\Message{b}{r}{q}}$} (q42);

  	\path (q20) edge node[left] {$\Action{!}{\Message{b}{r}{q}}$} (q30);

    \path (q21) edge node[left] {$\Action{!}{\Message{b}{r}{q}}$} (q31);

  	\path (q40) edge node[above] {$\Action{!}{\Message{c}{p}{q}}$} (q30) ;

  	 \path (q41) edge node[above] {$\Action{!}{\Message{c}{p}{q}}$} (q31) ;

  	\path (q42) edge node[above] {$\Action{!}{\Message{c}{p}{q}}$} (q32) ;

  	\path (q43) edge node[above] {$\Action{!}{\Message{c}{p}{q}}$} (q33) ;

  	\path (q50)  edge  node[above] {$\Action{!?}{\Message{a}{p}{r}}$} (q41);
  	\path (q50) edge node[above] {$\Action{!}{\Message{a}{p}{r}}$} (q40);

    \path (q53) edge node[above] {$\Action{!}{\Message{a}{p}{r}}$} (q43);
  	\path (q53)  edge node[above] {$\Action{!?}{\Message{a}{p}{r}}$} (q42);

    \path (q21) edge node [left] {$\Action{!?}{\Message{b}{r}{q}}$} (q32);
    \path (q20) edge node [left] {$\Action{!?}{\Message{b}{r}{q}}$} (q33);

  \end{scope}
\end{tikzpicture}
		\caption{System $\system_1$ and Automaton $\AutSRauto$ }\label{figure:ex_aut_SR}
		\end{center}
	\end{figure}

\end{example}
\begin{restatable}{lemma}{lemmaautSRcorrect}\label{lemma:aut_SR_correct}
$w \in \LangSRauto$ for some $\globalstatelj$ iff $\globalstateli \trMSC{msc(w)} \globalstatelf$.
\end{restatable}

\begin{proof}
Observe that, by construction of $\Aut{SR}\triple$,
$\globalstate{l^0_{\Aut{SR}}}\Tr{w}
(\globalstate{l},\globalstate{l'})$
iff
$\globalstate{in}\Tr{\sigma_1(w)}\globalstate{l}$
and
$\globalstate{mid}\Tr{\sigma_2(w)}\globalstate{l'}$ (this can
be shown by an easy induction on the length of $w$).
In particular, $w$ is accepted iff
$\globalstate{in}\Tr{\sigma_1(w)}\globalstate{mid}$
and
$\globalstate{mid}\Tr{\sigma_2(w)}\globalstate{fin}$,
which is equivalent to
$\globalstate{in}\trMSC{msc(w)}\globalstate{fin}$.\qed
\end{proof}

%
%

\subsection{Automata of  causal delivery exchanges}


Let us now move to the trickier part, namely the recognition
of words coding MSCs that satisfy causal delivery.
Let $\msc=(Ev,\lambda,\prec_{po},\prec_{src})$ be an MSC,
and  $v\in\lambda^{-1}(S)$  a send event, we write $ev_S(v)$ for the event $v$ and, when it exists, $ev_R(v)$ for the
event $v'\in\lambda^{-1}(R)$ such that $v\prec_{src} v'$. We say that $v$ is unmatched

\noindent
\begin{minipage}[c]{6.5cm}
  if $ev_R(v)$ is undefined. We recall from \cite{DBLP:conf/cav/BouajjaniEJQ18}
the notion of conflict graph.  Intuitively, it captures some (but not all) causal
dependencies between events.
 The figure on the right represents an MSC and its associated conflict graph.
\end{minipage}
\begin{minipage}[c]{5.5cm}
  \hspace*{0.3cm}
  \vspace*{-0.2cm}
\begin{tikzpicture}[baseline, scale = 0.8]
\begin{scope}[shift = {(6.2,0)}]
    \coordinate(pa) at (0,-0.25) ;
    \coordinate (pb) at (0,-2.25) ;
    \coordinate (qa) at (1,-0.25) ;
    \coordinate (qb) at (1,-2.25) ;
    \coordinate (ra) at (2,-0.25) ;
    \coordinate (rb) at (2,-2.25) ;
    \draw (0,0) node{$p$} ;
	\draw (1,0) node{$q$} ;
	\draw (2,0) node{$r$} ;
    \draw (pa) -- (pb) ;
    \draw (qa) -- (qb) ;
    \draw (ra) -- (rb) ;

    \coordinate (s1) at (1,-0.75);
    \coordinate (r1) at (0,-0.75);
    \draw[>=latex,->] (s1) -- node [above,sloped] {$\amessage_1$} (r1);
    \coordinate (s2) at (0, -1.85);
    \coordinate (r2) at (1, -1.85);
    \coordinate (s3) at (1, -1.3);

    \coordinate (r3) at (2, -1.3);

    \draw[>=latex,->] (s2) -- node [above,sloped] {$\amessage_2$} (r2);
    \draw[>=latex,->] (s3) -- node [above,sloped] {$\amessage_3$} (r3);
\end{scope}
\begin{scope}[shift = {(9.5, -1.75)}]
	\node[draw] (m1) at (0,0) {$v_1$};
	\node[draw] (m2) at (2,0) {$v_2$};
	\node[draw] (m3) at (0,1.55) {$v_3$};

	\draw[->] (m1) -- node [above] {RS} node [below] {SR} (m2);
	\draw[->] (m1) -- node [left] {SS} (m3);
	\draw[->] (m3) -- node [above,sloped] {SR} (m2);
\end{scope}
\end{tikzpicture}
\vspace*{-0.1cm}
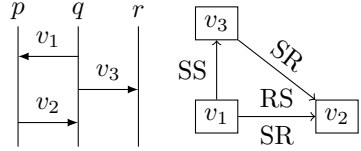
\captionof{figure}{MSC $\mu_3$ and its conflict graph}
\end{minipage}
\begin{definition}[Conflict Graph]\label{def:conflict-graph}
  The conflict graph $\cgraph{\msc}$ of an MSC $\msc=(Ev,\lambda,\prec_{po},\prec_{src})$ is the labeled
  graph $(V,\{\cgedge{XY}\}_{X,Y\in\{R,S\}})$ where $V=\lambda^{-1}(S)$,
  and for all $v,v'\in V$,
  there is a $XY$ dependency edge $v\cgedge{XY}v'$ between $v$ and $v'$ $(X,Y\in\{S,R\})$,
  if $ev_X(v)$ and $ev_Y(v')$ are defined and $ev_X(v)\prec_{po}ev_Y(v')$.
\end{definition}

\noindent
The \emph{extended conflict graph}~\cite{DBLP:conf/fossacs/GiustoLL20} $\ecgraph{\msc}$ is obtained by adding
all dashed edges $v\cgedgeD{XY}v'$ satisfying the relation $\cgedgeD{XY}$
in Fig.~\ref{fig:conflict-graph-extension}.
Intuitively, $v \cgedgeD{XY}v'$ expresses that
the event $X$ of $v$ must happen before  the event  $Y$ of $v'$ due to: their
order on the same machine (Rule 1), or the fact that a send
happens before its matching receive (Rule 2), or  to the mailbox
semantics (Rules 3 and 4), or because of a chain of such dependencies (Rule 5). This captures all
constraints induced by the mailbox communication, and it has been shown that an  MSC satisfies causal delivery if and only if its
extended conflict graph is acyclic (Theorem 2 in  \cite{DBLP:conf/fossacs/GiustoLL20}).

\begin{figure}[t]
\begin{center}
\begin{minipage}{.33\columnwidth}
\begin{center}
\AxiomC{\small $v_1 \cgedge{XY} v_2$}
\LeftLabel{\small (Rule 1)}
\UnaryInfC{\small $v_1 \cgedgeD{XY} v_2$}
\DisplayProof
\end{center}
\end{minipage}
\begin{minipage}{.33\columnwidth}
\begin{center}
\AxiomC{\small$v$ is matched}
\LeftLabel{\small(Rule 2)}
\UnaryInfC{\small$v \cgedgeD{SR} v$}
\DisplayProof
\end{center}
\end{minipage}
\begin{minipage}{.33\columnwidth}
\begin{center}
\AxiomC{\small$v_1 \cgedge{RR} v_2$}
\LeftLabel{\small(Rule 3)}
\UnaryInfC{\small $v_1 \cgedgeD{SS} v_2$}
\DisplayProof
\end{center}
\end{minipage}
\end{center}

\begin{center}
\begin{minipage}{.5\columnwidth}
\begin{center}
 \AxiomC{\small
    \begin{tabular}{c}
	$v_1$ is matched
    \qquad
    $v_2$ is unmatched
    \\
    $\receiver{v_1}=\receiver{v_2}$
  \end{tabular}}
  \LeftLabel{\small(Rule 4)}
\UnaryInfC{\small$v_1 \cgedgeD{SS} v_2$}
\DisplayProof
\end{center}
\end{minipage}
\begin{minipage}{.5\columnwidth}
\begin{center}
\AxiomC{\small$v_1 \cgedgeD{XY}\cgedgeD{YZ} v_2$}
\LeftLabel{\small(Rule 5)}
\UnaryInfC{\small$v_1 \cgedgeD{XZ} v_2$}
\DisplayProof
\end{center}
\end{minipage}
\end{center}

 \vspace*{-0.5cm}
\caption{\label{fig:rules-cgedgeD}Deduction rules for extended dependency edges of the conflict graph \vspace*{-0.5cm}}
\label{fig:conflict-graph-extension}
\end{figure}

We  build an automaton that recognizes the words $w$
such that $msc(w)$ satisfies causal delivery.
To this aim, we associate to each MSC a ``buffer state'' that
contains enough information to determine whether its extended conflict
graph is acyclic.
%
We write $\absBrSet$ for the set $(2^{\procSet} \times 2^{\procSet} )^{\procSet}$. The \emph{buffer state} $\absBr(\msc)\in\absBrSet$
of the MSC $\msc$ is the tuple
$\absBr(\msc)=(\mathcal{C}^{\msc}_{S,p}, \mathcal{C}^{\msc}_{R,p})_{p\in \procSet}$
such that for all $p\in\procSet$:
$$
\begin{array}{ll}
\Crxp{S}{p}{\msc}= &\{\sender{v}\mid v'\cgedgeD{SS}v~\&~ v'\mbox{ is unmatched }\&~\receiver{v'}=p\}~ \cup\\
& \{\sender{v}\mid  v\mbox{ is unmatched} ~\&~\receiver{v}=p \}\\
\Crxp{R}{p}{\msc}= &\{\receiver{v}\mid v'\cgedgeD{SS}v~\&~ v'\mbox{ is unmatched }\&~ \receiver{v'}=p~\&~v\mbox{ is matched}\}
\end{array}
$$
We can show that the $\ecgraph{\msc}$ is acyclic if for all $p \in \procSet$, $p\not\in \Crxp{R}{p}{\msc}$ (im\-mediate consequence of Theorem 2 in \cite{DBLP:conf/fossacs/GiustoLL20}). Moreover, we write
$\absBrSetGood$ for the subset of
$\absBrSet$ formed by the tuples $(C_{S,p},C_{R,p})_{p\in\procSet}$
such that $p\not\in C_{R,p}$ for all $p$.

\begin{proposition}[\cite{DBLP:conf/fossacs/GiustoLL20}]
  \label{prop:causaldeliv}
For $w \in \Sigma^*$, $msc(w)$ satisfies causal delivery if and only if  $ \absBr(\msc(w)) \in \absBrSetGood$.
\end{proposition}
%

%
%
%

Noticing
that $\absBrSet$ is finite, we build an automaton
$\cdauto{\Br_0}{\Br_1}$ with $\Br_0,\Br_1 \in \absBrSet$.
The intuition behind these two buffer states
is that $\Br_0$ summarises the conflict graph
derived from previous exchanges and $\Br_1$ summarises the
conflict graph obtained when a new exchange is added.


\begin{definition}[Automaton of causal exchanges] \label{def:feasibleuto}
The automaton $\cdauto{\Br_0}{\Br_1} $ is defined as follows:
\begin{itemize}
\item $\absBrSet$ is the set of states,
\item $\Br_0$ is the initial state (hereafter, we assume that $\Br_0 = (C^{(0)}_{S,p},C^{(0)}_{R,p})_{p\in \procSet}$).
\item $\{\Br_1$\} is the set of final states
\item the transition relation $(\tr{a})_{a\in\Sigma}$ is defined as follows:
\begin{itemize}
\item $(C_{S,p},C_{R,p})_{p\in \procSet} \tr{!?\amessage^{p\to q}}(C'_{S,p},C'_{R,p})_{p\in \procSet}$ holds if for all $r\in\procSet$: let the intermediate set $C_{S,r}''$ be defined by
$$
C_{S,r}''=\left\{\begin{array}{l}
C_{S,r}\cup\{p\} \mbox{ if }p\in C_{R,r}^{(0)}\mbox{ or }q\in C_{R,r}\\
C_{S,r}\mbox{ otherwise}
\end{array}\right.
$$
Then
$$
\hspace*{-1cm}
C_{S,r}'=\left\{\begin{array}{l}
C_{S,r}''\cup C_{S,q} \mbox{ if }p\in C_{S,r}''\\
C_{S,r}\mbox{ otherwise}
\end{array}\right.
\quad\mbox{and}\quad
C_{R,r}'=\left\{\begin{array}{l}
C_{R,r}\cup \{q\} \cup C_{R,q} \mbox{ if }p\in C_{S,r}''\\
C_{R,r}\mbox{ otherwise}
\end{array}\right.
$$
\item $(C_{S,p},C_{R,p})_{p\in \procSet} \tr{!\amessage^{p\to q}}(C'_{S,p},C'_{R,p})_{p\in \procSet}$ holds if for all $r\in\procSet$,
$$
C_{S,r}'=\left\{\begin{array}{l}
C_{S,r}\cup \{p\} \mbox{ if }q=r\mbox{ or } q\in C_{R,r}\\
C_{S,r}\mbox{ otherwise}
\end{array}\right.
\quad\mbox{and}\quad
C_{R,r}'=C_{R,r}
$$

\end{itemize}
\end{itemize}
Let $\mathcal{L}(\Br_0, \Br_1)$ denote
the language recognized by $\cdauto{\Br_0}{\Br_1}$.
\end{definition}

\begin{example}
Consider $\mu_4=msc(w)$ with $w=!\amessage_3^{p_1\to p_2}!?\amessage_4^{p_3\to p_2}!?\amessage_5^{p_4\to p_6}!?\amessage_6^{p_6\to p_7}$
\\ \noindent
\begin{minipage}[c]{6cm}
and assume we start with $\Br_{0}$ such that $C_{S,p_5}=\{p_4\}$ and $C_{R,p_5}=\{p_3\}$.
Then the update of $\Br$ (or, more precisely, of $C_{S,p_5}$, $C_{R,p_5}$, and $C_{S,p_2}$) after reading each message is shown below. Note how $\amessage_6$
has no effect, despite the fact that $p_6\in C_{R,p_5}$ at the time the message is read.
\end{minipage}
\begin{minipage}[c]{7cm}
  \vspace*{-0.3cm}
  \hspace*{0.3cm}
\begin{tikzpicture}[baseline]
\node at (0,-0.5) {$p_1$};
\node at (0.85,-0.5) {$p_2$};
\node at (1.7,-0.5) {$p_3$};
\node at (2.55,-0.5) {$p_4$};
\node at (3.4,-0.5) {$p_5$};
\node at (4.25,-0.5) {$p_6$};
\node at (5.1,-0.5) {$p_7$};
\draw (0,-.75) -- (0,-3.25);
\draw (0.85,-.75) -- (0.85,-3.25);
\draw (1.7,-.75) -- (1.7,-3.25);
\draw (2.55,-.75) -- (2.55,-3.25);
\draw (3.4,-.75) -- (3.4,-3.25);
\draw (4.25,-.75) -- (4.25,-3.25);
\draw (5.1,-.75) -- (5.1,-3.25);
\draw[dashed,->,>=latex] (2.55,-1.25) -- node[above]{$\amessage_1$} (3.4,-1.25);
\draw[->,>=latex] (2.55,-1.75) -- node[above] {$\amessage_2$} (1.7,-1.75);
\draw[dashed] (-0.5,-2) -- (5.5,-2);
\draw[dashed,->,>=latex] (0,-2.5) -- node[above] {$\amessage_3$} (0.85,-2.5);
\draw[->,>=latex] (1.7,-3) -- node [above] {$\amessage_4$} (0.85,-3);
\draw[->,>=latex] (2.55,-2.5) -- node [sloped,above, pos=0.3] {$\amessage_5$} (4.25,-3);
\draw[->,>=latex] (4.25,-2.5) -- node [above] {$\amessage_6$} (5.1,-2.5);
\end{tikzpicture}
\vspace*{-0.3cm}
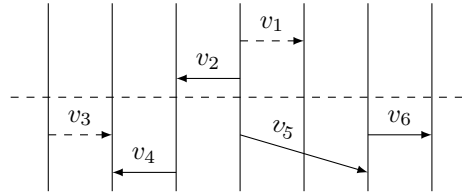
\captionof{figure}{MSC $\mu_4$}
\end{minipage}

\vspace{0.3cm}
\begin{tabular}{l|ccccccccc}
$C_{S,p_5}$ & $\{p_4\}$ && $\{p_4\}$ && $\{p_1,p_3,p_4\}$ && $\{p_1,p_3,p_4\}$ && $\{p_1,p_3,p_4\}$ \\
$C_{R,p_5}$ & $\{p_3\}$ & $\tr{!\amessage_3}$ & $\{p_3\}$ &
$\tr{!?\amessage_4}$ & $\{p_2,p_3\}$ & $\tr{!?\amessage_5}$ & $\{p_2,p_3,p_6\}$
& $\tr{!?\amessage_6}$ & $\{p_2,p_3,p_6\}$ \\
$C_{S,p_2}$ & $\emptyset$ && $\{p_1\}$ && $\{p_1\}$ && $\{p_1\}$ && $\{p_1\}$
\end{tabular}
\end{example}

Next lemma states that $\anauto(\Br,\Br')$
recognizes the words $w$ such that $msc(w)$,
starting with an initial buffer state $\Br$,
ends in final buffer state $\Br'$.

\begin{restatable}{lemma}{lemmacausaldel}\label{lemma:language_B}
Let $\Br, \Br' \in \absBrSet$ and $w\in\Sigma^*$. Then
$w\in \mathcal{L}(\Br, \Br')$ if and only if for all MSC
$\msc$ such that $\Br = \absBr(\mu)$,
$\Br' = \absBr(\mu \cdot msc(w))$.
\end{restatable}
\begin{proof}
Take $w = a_0 \dots a_n \in \Sigma^*$. To prove the lemma it is sufficient to show that if    $\Br = \absBr(\mu)$ and :
$$\Br = (C^{(0)}_{S,p}, C^{(0)}_{R,p})_{p \in \procSet} \tr{a_{0}}  (C^{(1)}_{S,p}, C^{(1)}_{R,p})_{p \in \procSet} \tr{a_{1}} \cdots \tr{a_n} (C^{(n+1)}_{S,p}, C^{(n+1)}_{R,p})_{p \in \procSet} = \Br',$$
then $\Br' = \absBr(\mu \cdot msc(w))$.

The proof proceeds by induction on the length of $w$ where the inductive hypothesis is that $(C^{(n)}_{S,p}, C^{(n)}_{R,p})_{p \in \procSet} = \absBr(\mu \cdot msc(a_0 \dots a_{n-1}))$.

We start by showing that $\forall r \in \procSet$, $C^{(n+1)}_{S,r} = \Crxpprime{S}{r}{\msc \cdot msc(w)}$.
Suppose that $p \in C^{(n+1)}_{S,r}$. If $p \in C^{(0)}_{S,r}$ then we can immediately conclude that $p \in \Crxpprime{S}{r}{\msc \cdot msc(w)}$ since $C^{(0)}_{S,r} = \Crxp{S}{r}{\msc} \subseteq C^{(n+1)}_{S,r}$ and function $\absBr(\cdot)$ is increasing monotone.
Instead if $p \notin C^{(0)}_{S,r}$ then the following can happen (without loss of generality suppose that $p$ has been added while reading the last symbol of $w$):
 \begin{itemize}
 \item $a_n=  !?\amessage^{p\to q}$ and $p \in C^{(0)}_{R, r}= \Crxp{R}{r}{\msc}$

 Then there exists a message in $\msc$, $\amessage'' $ such that $\receiver{v''} = p$ and $v'\cgedgeD{SS}v''$ with $v'$ unmatched. Then it is easy to see that $v''\cgedgeD{SS}v$ and therefore $p \in \Crxpprime{S}{r}{\msc \cdot msc(w)}$;

 \item $a_n= !?\amessage^{p\to q}$ and $q \in C^{(n)}_{R, r} = \Crxp{R}{r}{\msc \cdot msc(a_0 \dots a_{n-1})}$

 Then there exists a message $\amessage'' $ in $\msc \cdot msc(a_0 \dots a_{n-1})$,  such that $\receiver{v''} = q$ and $v'\cgedgeD{SS}v''$ with $v'$ unmatched. Then it is easy to see that $v''\cgedgeD{SS}v$ and therefore $p \in \Crxpprime{S}{r}{\msc \cdot msc(w)}$;

 \item $a_n=  !?\amessage^{p'\to q}$ and $p' \in C''^{(n)}_{S, r}$ and $p \in C^{(n)}_{S,q}= \Crxp{S}{q}{\msc \cdot msc(a_0 \dots a_{n-1})}$

  Then there exists a message $\amessage'' $ in $\msc \cdot msc(a_0 \dots a_{n-1})$,  such that $\sender{v''} = p$ and $v'\cgedgeD{SS}v''$ with $v'$ unmatched and $\receiver{v'} = q$. Then it is easy to see that $v\cgedgeD{SS}v''$ and since $p' \in C''^{(n)}_{S, r}$ and with an analysis similar to the one above we have $v'''\cgedgeD{SS}v$ with $v'''$ unmatched and  we can conclude $p \in \Crxpprime{S}{r}{\msc \cdot msc(w)}$.

 \item $a_n=  !\amessage^{p\to q}$ and $p \in C_{R, r}$

 Analogous to case 2 above.

\item $a_n= !\amessage^{p\to r}$

We can immediately conclude $p \in \Crxpprime{S}{r}{\msc \cdot msc(w)}$ as $v$ is unmatched and $\receiver{v}=r$.

 \end{itemize}

Now suppose that $p \in \Crxpprime{S}{r}{\msc \cdot msc(w)}$ (without loss of generality we can assume $p \notin \Crxpprime{S}{r}{\msc \cdot msc(a_0 \dots a_{n-1})}$). Then either  $a_n = ! \amessage ^{p\to r}$ then it is immediate to see that $p \in C^{(n+1)}_{S,r}$ or  $a_n = !? \amessage ^{p\to q}$ and $v'\cgedgeD{SS}v$ for some $v'$ unmatched and $\receiver{v'} = r$.
  $p \notin \Crxpprime{S}{r}{\msc \cdot msc(a_0 \dots a_{n-1})}$ entails that either $q\in  \Crxpprime{R}{r}{\msc \cdot msc(a_0 \dots a_{n-1})}$ or $p\in  \Crxpprime{R}{r}{\msc}$ (notice that since $w$ is an exchange, $p\notin  \Crxpprime{R}{r}{\msc \cdot msc(a_0 \dots a_{n-1}) } \setminus  \Crxpprime{R}{r}{\msc}$ ). In both cases we can conclude $p \in C^{(n+1)}_{S,r}$.

Next we show that  $\forall r \in \procSet$, $C'_{R,r} = \Crxpprime{R}{r}{\msc \cdot msc(w)}$.
Suppose that $p \in C^{(n+1)}_{R,r}$ (without loss of generality we can assume $p \notin C{(n+1)}_{R,r}$). This entails that:
\begin{itemize}
\item  either $a_n = !? \amessage ^{q \to p}$ with $q \in C''^{(n+1)}_{S,r} =  \Crxpprime{S}{r}{\msc \cdot msc(w)}$:

Then there exists a message $\amessage'' $ in $\msc \cdot msc(w)$,  such that $\sender{v''} = q$ and $v'\cgedgeD{SS}v''$ with $v'$ unmatched and $\receiver{v'} = r$. Then it is easy to see that $v''\cgedgeD{SS}v$  and  we can conclude $p \in \Crxpprime{R}{r}{\msc \cdot msc(w)}$;

\item or  $a_n = !? \amessage ^{q \to p'}$ with $q \in C''^{(n+1)}_{S,r}= \Crxpprime{S}{r}{\msc \cdot msc(w)}$ and $p \in C^{(n)}_{R,p'} = \Crxpprime{R}{p'}{\msc \cdot msc(a_0 \dots a_{n-1})}$:

Then there exists a message $\amessage'' $ in $\msc \cdot msc(w)$,  such that $\sender{v''} = q$ and $v'\cgedgeD{SS}v''$ with $v'$ unmatched and $\receiver{v'} = r$. Similarly there is $\amessage''' $ in $\msc \cdot msc(a_0 \dots a_{n-1})$ such that $\receiver{v'''} = p$ and $v^{iv}\cgedgeD{SS}v'''$ with $v^{iv}$ unmatched and $\receiver{v^{iv}} = p'$. Now when adding $v$ to the conflict graph we have  $v\cgedgeD{SS}v^{iv}$. Hence we can conclude $p \in \Crxpprime{R}{r}{\msc \cdot msc(w)}$.
\end{itemize}

Now suppose that $p \in \Crxpprime{R}{r}{\msc \cdot msc(w)}$ (without loss of generality we can assume $p \notin \Crxpprime{R}{r}{\msc \cdot msc(a_0 \dots a_{n-1})}$). We know $\Crxpprime{R}{r}{\msc \cdot msc(a_0 \dots a_{n-1})} = C^{(n)}_{R,r}$ and let $a_n= !? \amessage ^{q \to p}$.
The following can happen:
$q \in  \Crxpprime{S}{r}{\msc \cdot msc(a_0 \dots a_{n-1})},$
or $q \in  \Crxpprime{R}{r}{\msc \cdot msc(a_0 \dots a_{n-1})}$,
 or $p \in  \Crxpprime{S}{r}{\msc \cdot msc(a_0 \dots a_{n-1})}$.
The last case is when  $a_n= !? \amessage ^{q' \to p'}$ and $p \in   \Crxpprime{R}{p'}{\msc \cdot msc(a_0 \dots a_{n-1})}$ and $q' \in \Crxpprime{S}{r}{\msc \cdot msc(a_0 \dots a_{n-1})}$.
For all this cases,  by inductive hypothesis and by Definition \ref{def:feasibleuto} we can conclude $p \in C^{n+1}_{R,r}$.\qed

%
%

\end{proof}

\subsection{Language of reachable exchanges}

The only thing that remains to do is to combine the previous automata 
to define one that recognizes the (codings of) reachable exchanges.
The language $\LangSRauto$ contains
arbitrary exchanges which do not necessarily satisfy
causal delivery.
Here comes into play the $\anauto(\Br,\Br')$ automata, 
where we take $\Br$ and $\Br' \in \absBrSetGood$ in
order to ensure causal delivery.

Let
$$ 
	\feaslanguage(\globalstateli, \globalstatelf, \Br, \Br' )
	\eqdef 
	\bigcup_{\globalstatelj \in L_{\system} } \LangSRauto \cap 
	\mathcal{L}(\Br, \Br').
$$
Intuitively, 
$\feaslanguage(\globalstateli, \globalstatelf, \Br, \Br' )$ 
is the language of (codings of) exchanges between 
global states $\globalstateli$ and  $\globalstatelf$ 
starting with an initial buffer state $\Br$ and ending 
in final buffer state $\Br'$; when moreover  
$\Br, \Br' \in \absBrSetGood$, these exchanges satisfy
causal delivery.

The last step is to combine causal delivery exchanges so 
that they can be performed by the system one after the other 
from the initial state $\globalstate{l_0}$. 
This motivates the definition
of the following set $\mathcal R$ of \emph{reachable languages}.
Let $\Br_\emptyset = (\emptyset, \emptyset)_{p \in \procSet}$.

\begin{definition}[Reachable languages]
Given a system 
$\system=(L_\system, \delta_\system, \globalstate{l_0})$, 
the set $\mathcal{R}$ of reachable languages is 
the least set of languages of the form \\
$\feaslanguage(\globalstateli, \globalstatelf,\Br_i, \Br_f)$
defined as follows.
\begin{enumerate}
	\item for any $\globalstate{l} \in L_{\system}$ and
$\Br \in \absBrSetGood $,
	$\feaslanguage(\globalstate{l_0},\globalstate{l},\Br_{\emptyset},\Br)$ is in $\mathcal{R}$
	\item for any $\globalstate{l_1},\globalstate{l_2},\globalstate{l_3}  \in L_{\system}$ and any $\Br_1,\Br_2,\Br_3\in \absBrSetGood$, if
	$\feaslanguage(\globalstate{l_1},\globalstate{l_2},\Br_1,\Br_2)\in \mathcal{R}$  and $\feaslanguage(\globalstate{l_1},\globalstate{l_2},\Br_1,\Br_2)\neq\emptyset$ then $\feaslanguage(\globalstate{l_2},\globalstate{l_3},\Br_2,\Br_3)\in \mathcal{R}$.
\end{enumerate}

\end{definition}

Then the union $\bigcup \mathcal R$ of all reachable languages
is equal to the language 
$\reachlanguage=\{w\in\Sigma^*\mid msc(w)\mbox{ is reachable}\}$. 
As a consequence, we get the following result.

\begin{restatable}{theorem}{thmreachable}\label{th:reachable}
	$\reachlanguage$ is a regular language and is accepted
	by an effective finite state automaton.
\end{restatable}


\begin{proof}
	$\Rightarrow$ $w \in \bigcup\mathcal{R}$ so there is
	a sequence of words $w_1, \cdots, w_n \in \Sigma^*$ such that $\forall 1 \leq
	j \leq n, w_j \in \feaslanguage(\globalstate{l_{j-1}}, \globalstate{l_j},
	\Br_{j-1}, \Br_j) \neq \emptyset$ with $\Br_{j} \in \absBrSetGood$, and there is also
	$\globalstatelf \in L_\system, \Br_f \in \absBrSetGood$ such that  $w \in
	\feaslanguage(\globalstate{l_n}, \globalstatelf, \Br_n, \Br_f)$.

	So, $w_j \in
	\LangSR{\globalstate{l_{j-1}}}{\globalstate{l'_j}}{\globalstate{l_j}}$ for
	$\globalstate{l'_j} \in L_\system$, and, by Lemma~\ref{lemma:aut_SR_correct},
	$\globalstate{l_{j-1}} \trMSC{msc(w_j)} \globalstate{l_j}$, and
	$\globalstate{l_n} \trMSC{msc(w)} \globalstatelf$.

	Moreover, for all $1 \leq j \leq n$, $w_j \in \mathcal{L}(\Br_{j-1}, \Br_j)$ and $\Br_j \in \absBrSetGood$.
	So,  each $msc(w_j)$ verifies causal delivery, and,  we can easily show by induction that  $msc(w_1) \cdots msc(w_n) \cdot msc(w)$ verifies causal delivery too. Finally, by Proposition~\ref{proposition:msc_trace_state}, $msc(w_1) \cdots msc(w_n) \cdot msc(w) \in asTr(\system)$ and so, $msc(w)$ is reachable.

$\Leftarrow$
 $msc(w)$ is reachable so there is $\mu_1 \cdots \mu_n$ a sequence of MSCs such that $
\mu_1 \cdots \mu_n \cdot msc(w) \in asTr(\system)$.
Suppose that $\mu_1 \cdots \mu_n = \varepsilon$, then, $\globalstate{l_0}
\trMSC{msc(w)} \globalstatelf$ for some $\globalstatelf$. Therefore, $w \in
\LangSR{\globalstate{l_0}}{\globalstate{l}}{\globalstatelf}$ for some
$\globalstate{l} \in L_\system$. As $msc(w) \in asTr(\system)$, $msc(w)$ verifies
causal delivery and so there is $\Br$ such that $w \in
\mathcal{L}(\Br_0,\Br')$ with $\Br' \in \absBrSetGood$.
Finally, we have that $w \in \feaslanguage(\globalstate{l_0}, \globalstatelf, \Br_0, \Br)$ and so $w \in \bigcup\mathcal{R}$.

Now, suppose that  $\mu_1 \cdots \mu_n \neq \varepsilon$. Then, there is a sequence  $w_1, \cdots, w_n \in \Sigma^*$ such that $msc(w_i) = \mu_i$, $1\leq i\leq n$, and we can suppose that $w_1, \cdots, w_n \in \reachlanguage$.
By Lemma~\ref{lemma:language_B}, there is $\Br = \absBr(\mu_1 \cdots \mu_n), \Br' =  \absBr(\mu_1 \cdots \mu_n \cdot msc(w)) \in \absBrSetGood$ such that $w \in \mathcal{L}(\Br, \Br')$.
Moreover, there is $\globalstateli, \globalstatelf$ such that $\globalstateli \trMSC{msc(w)} \globalstatelf$, so $w \in \LangSR{\globalstateli}{\globalstatelj}{\globalstatelf}$ for some $\globalstatelj$.
Finally, we have $w \in \feaslanguage(\globalstateli, \globalstatelf, \Br, \Br')$ and then $w \in \reachlanguage$.\qed
\end{proof}

%
%
%

\section{Prime exchanges}\label{sec:prime}


We reformulate the primality of an exchange in terms of its conflict graph. 
We say that
the conflict graph $\cgraph{\msc}$ associated with the MSC $\msc$ is strongly connected if
for all $v,v'\in V$ it holds that $v\to^* v'$, where $\to^*$ is the reflexive transitive closure of
$\to=\bigcup_{X,Y\in\{S,R\}}\cgedge{XY}$.

\begin{restatable}{lemma}{lemmaprimestronglyconnected}\label{lem:prime-strongly-connected}
An exchange $\msc$ is prime iff $\cgraph{\msc}$ is strongly connected.
\end{restatable}
\begin{proof}
Let $\msc=\msc_1\cdots\msc_n$ be an MSC formed with a sequence
of exchanges. Let $e,e'$ be two events of $\msc$, and let $i,i'\in\{1,\ldots,n\}$ be such
that $e$ appears in $\msc_{i}$ and $e'$ appears in $\msc_{i'}$.
If there is an edge $e\xrightarrow{XY} e'$ in the conflict graph of $\msc$,
then $i\leq i'$. As a consequence, if $e$ and $e'$ are on a same
strongly connected component, then $i=i'$, and if the conflict graph
of $\msc$ is strongly connected, then $n=1$ and $\msc$ is a prime
exchange. \qed
\end{proof}

Next we  discuss the construction of the automaton that recognizes $\{w\in \Sigma^*\mid msc(w) \mbox{ is prime}\}$.
Since there are infinitely many $\cgraph{msc(w)}$, in order to have a finite state automaton, we compute a finite abstractions
of $\cgraph{msc(w)}$ 
that is sound in the sense
that $\cgraph{msc(w)}$ is strongly connected if and only if its abstraction is of a certain shape.
Let us now define this abstraction.

We need to define some graph transformations. The graphs we are going to manipulate are
oriented graphs labeled with a pair of set of processes on each vertex. We call such objects P-graphs.
Formally, a P-graph is a tuple $(V,E,\lambda_S,\lambda_R)$ with $E\subseteq V\times V$ and
$\lambda_X:V\to 2^{\procSet}$ for $X \in \{S,R\} $. The P-graph $\pgraphof{\msc}$ associated with the conflict graph
$\cgraph{\msc}=(V,\{\cgedge{XY}\}_{X,Y\in\{S,R\}})$ is $(V,E,\lambda_S,\lambda_R)$ where
(1) $(v,v')\in E$ if $v\cgedge{XY}v'$ for some $X,Y$, (2) $\lambda_S(v)=\{\sender{v}\}$, and
(3) if $v$ is matched, then $\lambda_R(v)=\{\receiver{v}\}$, and if $v$ is unmatched
$\lambda_R(v)=\emptyset$.

The first graph transformation we consider consists in merging the vertices that belong to a same strongly
connected component (SCC). Formally, let $G=(V,E,\lambda_S,\lambda_R)$ be a P-graph, and let
$\mergeof{G}=(V',E',\lambda_S,\lambda_R)$ be defined by (1) $V'$ is the set of maximal SCCs of $G$,
(2) for two distinct maximal SCCs $U,U'$, $(U,U')\in E'$ if there are $v\in U$ and $v'\in U'$ such that
$(v,v')\in E^+$ (the transitive closure of $E$), (3) for $X=S,R$, $\lambda_X(U)=\bigcup_{v\in U}\lambda_X(v)$.

The second graph transformation we consider consists in erasing some of the processes that appear in the labels.
Let $G=(V,E,\lambda_S,\lambda_R)$ be a fixed P-graph, and let $v\in V$, $X\in\{S,R\}$, and $p\in\lambda_X(v)$
be fixed. We say that $p$ is X-redundant in $v$ if there are $v_1,v_2$ such that (1) $(v_1,v)\in E^+$ and
$(v,v_2)\in E^+$, and (2) $p\in\lambda_X(v_1)\cap\lambda_X(v_2)$. Intuitively, $p$ is redundant in $v$ if it also
appears in the label of an ancestor and a descendant of $v$. We define the P-graph $\eraseof{G}$ as $(V,E,\lambda_S',\lambda_R')$ where for all $X\in\{S,R\}$, for all $v\in V$, $\lambda_X'(v)$ is the set of processes
$p\in\lambda_X(v)$ such that $p$ is not X-redundant at $v$.

The last graph transformation we consider consists in sweeping out the vertices labeled with empty sets of
processes. Formally, for $G=(V,E,\lambda_S,\lambda_R)$, the P-graph $\sweepof{G}$ is $(V',E',\lambda_S,\lambda_R)$
where $V'=\{v\in V\mid \lambda_S(v)\cup\lambda_R(v)\neq\emptyset\}$ and $E'=E\cap V'\times V'$. The abstraction
$\alpha(G)$ of a P-graph $G$ is defined as $\sweepof{\eraseof{\mergeof{G}}}$.
An example of the construction is in Fig~\ref{fig:example-abstraction}.

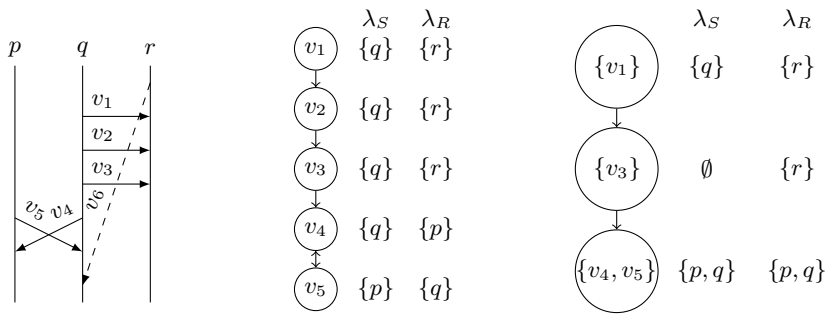
\begin{figure}[t]
\begin{tikzpicture}
\begin{scope}[xshift=0, scale = 0.9]

  \coordinate(pa) at (0,-0.25) ;
  \coordinate (pb) at (0,-3.75) ;
  \coordinate (qa) at (1,-0.25) ;
  \coordinate (qb) at (1,-3.75) ;
  \coordinate (ra) at (2,-0.25) ;
  \coordinate (rb) at (2,-3.75) ;
  \draw (0,0) node{$p$} ;
\draw (1,0) node{$q$} ;
\draw (2,0) node{$r$} ;
  \draw (pa) -- (pb) ;
  \draw (qa) -- (qb) ;
  \draw (ra) -- (rb) ;

  \draw[>=latex,->] (1,-1) -- node [above,sloped, pos = 0.3] {$\amessage_1$} (2,-1);
  \draw[>=latex,->] (1,-1.5) -- node [above,sloped, pos = 0.3] {$\amessage_2$} (2,-1.5);
  \draw[>=latex,->] (1,-2) -- node [above,sloped, pos = 0.3] {$\amessage_3$} (2,-2);
  \draw[>=latex,->] (1,-2.5) -- node [above,sloped, pos = 0.2] {$\amessage_4$} (0,-3);
  \draw[>=latex,->] (0,-2.5) -- node [above,sloped, pos = 0.2] {$\amessage_5$} (1,-3);
  \draw[>=latex,->, dashed] (2,-.5) -- node [above,sloped, pos = 0.6] {$\amessage_6$} (1,-3.5);

\end{scope}
\begin{scope}[xshift=4cm, scale = 0.8]
\node at (1,0.5) {$\lambda_S$};\node at (2,0.5) {$\lambda_R$};
\node[circle,draw, text width=.2cm] (1) at (0,0) {$\hspace*{-0.05cm}\amessage_1$}; \node at (1,0) {$\{q\}$}; \node at (2,0) {$\{r\}$};
\node[circle,draw, text width=.2cm] (2) at (0,-1) {$\hspace*{-0.05cm}\amessage_2$}; \node at (1,-1) {$\{q\}$}; \node at (2,-1) {$\{r\}$};
\node[circle,draw, text width=.2cm] (3) at (0,-2) {$\hspace*{-0.05cm}\amessage_3$}; \node at (1,-2) {$\{q\}$}; \node at (2,-2) {$\{r\}$};
\node[circle,draw, text width=.2cm] (4) at (0,-3) {$\hspace*{-0.05cm}\amessage_4$}; \node at (1,-3) {$\{q\}$}; \node at (2,-3) {$\{p\}$};
\node[circle,draw, text width=.2cm] (5) at (0,-4) {$\hspace*{-0.05cm}\amessage_5$}; \node at (1,-4) {$\{p\}$}; \node at (2,-4) {$\{q\}$};
\draw[->] (1) -- (2);
\draw[->] (2) -- (3);
\draw[->] (3) -- (4);
\draw[<->] (4) -- (5);
\end{scope}
\begin{scope}[xshift=8cm, scale=0.8]
\node at (1.5,0.5) {$\lambda_S$};
\node at (3,0.5) {$\lambda_R$};
\node[circle,draw, text width=.75cm] (n0) at (0,-.3) {$\hspace*{0.07cm}\{\amessage_1\}$};
\node at (1.5,-.3) {$\{q\}$};
\node at (3,-.3) {$\{r\}$};
\node[circle,draw, text width=.75cm] (n1) at (0,-2) {$\hspace*{0.07cm}\{\amessage_3\}$};
\node at (1.5,-2) {$\emptyset$};
\node at (3,-2) {$\{r\}$};
\node[circle,draw, text width=.75cm] (n2) at (0,-3.7) {$\hspace*{-0.17cm}\{\amessage_4,\amessage_5\}$};
\node at (1.5,-3.7) {$\{p,q\}$};
\node at (3,-3.7) {$\{p,q\}$};
\draw[->] (n0) -- (n1); \draw[->] (n1) -- (n2);
\end{scope}
\end{tikzpicture}
\caption{\label{fig:example-abstraction}MSC $\msc_5$, 
its associated P-graph $\pgraphof{\msc_5}$, and the abstraction $\alpha(\pgraphof{\msc_5})$.}
\end{figure}

\begin{restatable}{lemma}{lemmasciffsinglevertex}\label{lem:sc-if-single-vertex}
$\cgraph{\msc}$ is strongly connected iff $\alpha(\pgraphof{\msc})$ is a single vertex graph.
\end{restatable}
\begin{proof}
By definition of $\alpha$, and in particular
of function $\mergeof{.}$,
a vertex of $\alpha(\pgraphof{\msc})$ corresponds to
a strongly connected component of $\cgraph{\msc}$.\qed
\end{proof}

By construction, for any process $p$, and for any $X\in\{S,R\}$,
there are at most two vertices $v$ of $\alpha(\pgraphof{\msc})$ such that
$p\in\lambda_X(v)$. From this, we deduce that $\alpha(\pgraphof{\msc})$ has at most
$2|\procSet|$ vertices, and as a consequence:
\begin{restatable}{lemma}{lemfiniteabstraction}\label{lem:finite-abstraction}
$\sharp \{\alpha(\pgraphof{\msc})\mid\msc\mbox{ is an exchange}\}\leq 2^{6|\procSet|^2}$.
\end{restatable}

\begin{proof}
Let $n\geq 0$ be fixed and let us give an upper bound on the
number of P-graphs with $n$ vertices. First, there are $2^{n(n-1)}$ different
possible choices for the edge relation.
By construction (in particular, by definition of function $\eraseof{.}$),
it holds that
$$
(1)\qquad
\forall p\in\procSet, \forall X\in\{S,R\},
\sharp\{v\mid p\in\lambda_X(v)\}\leq 2.
$$
A choice for the
$\lambda$ function is therefore the choice, for each $p$, of at most two
vertices $v$ such that $p\in\lambda_S(v)$ and at most two other vertices
$v$ such that $p\in\lambda_R(v)$. So there are at most $n^4$ different
choices for each $p$, and at most $n^{4|\procSet|}$ different
choices for $\lambda$. To sum up, there are less than
$2^{n^2}n^{4|\procSet|}$ P-graphs with $n$ vertices.

Now, again from (1), there are at most $2|\procSet|$ vertices in a P-graph,
so the number of P-graph is bounded by
$$
\sum_{n=1}^{2|\procSet|}2^{n^2}n^{4|\procSet|}\leq
2|\procSet|2^{|\procSet|^2}(2|\procSet|)^{4|\procSet|}
\leq 2^{|\procSet|^2}2^{|\procSet|^2}(2^{|\procSet|})^{4|\procSet|}=2^{6|\procSet|^2}.
$$
\qed
\end{proof}

\vspace{1cm}
There are therefore finitely many $\alpha(\pgraphof{\msc})$.
This allows us to define
the automaton that computes $\alpha(\pgraphof{msc(w)})$ for any $w\in \Sigma^*$
and
 accepts $w$ in the language of this new automaton
when this P-graph is a single vertex graph.
Let $G=(V,E,\lambda_S,\lambda_R)$ and
a letter $\dag\amessage^{p\to q}\in\Sigma$ be fixed.
We want to define the transition function $\delta_g$ of our
automaton, or in other words, the P-graph
$\delta_g(G,\dag\amessage^{p\to q})$
reached after adding the message
$\dag\amessage^{p\to q}$ to the MSC.
%
We let $\delta_g(G,\dag\amessage_0^{p\to q})=\alpha(G')$.
$G'=(V',E',\lambda_S',\lambda_R')$ is defined as follows: (1) $V'=V\uplus\{v_0\}$, (2) $\lambda_S'(v_0)=\{p\}$,
(3) if $\dag=!?$, then $\lambda_R'(v_0)=\{q\}$, and if $\dag=!$, then $\lambda_R'(v_0)=\emptyset$, (4) for all
$v\in V$, for all $X\in\{S,R\}$, $\lambda_X'(v)=\lambda_X(v)$, and (5) the set of edges $E'$ is defined as

\noindent
\begin{minipage}[c]{8.5cm}
  $
  \begin{array}{lll}
  \quad E'&=&E\cup\{(v,v_0)\mid p\in\lambda_S(v)\}
  \cup\{(v_0,v)\mid p\in\lambda_R(v)\}\\
  &\cup&\left\{\begin{array}{ll}
  \{(v,v_0)\mid q\in\lambda_S(v)\cup\lambda_R(v)\} & \mbox{if }\dag=!?\\
  \emptyset & \mbox{if }\dag=!
  \end{array}\right.
  \end{array}
  $
  For example, consider the MSC $\msc$ of
  Fig.~\ref{fig:example-abstraction}
  and let $G=\alpha(\pgraphof{\msc})$ be its associated
  abstracted P-graph. Let $G'$ be defined as above while
  reading $!?\amessage_6^{r\to q}$.
  Then $G'$ is the graph on the right, and
  $\delta_g(G,\amessage_6^{r\to q})$ is a single vertex graph.
\end{minipage}
\hspace*{0.5cm}
\begin{minipage}[c]{3cm}
  \begin{center}
  \begin{tikzpicture}[scale= 0.8, every node/.style={scale=0.7}]
  \node[circle,draw, text width =0.75cm ] (n0) at (0,0) {$\hspace*{0.05cm}\{\amessage_1\}$};
  \node[circle,draw, text width = 0.75cm] (n1) at (2.5,0) {$\hspace*{0.05cm}\{\amessage_6\}$};
  \node[circle,draw, text width = 0.75cm] (n2) at (0,-2) {$\hspace*{0.05cm}\{\amessage_3\}$};
  \node[circle,draw, text width = 0.75cm] (n3) at (2.5,-2) {$\hspace*{-0.21cm}\{\amessage_4,\amessage_5\}$};
  \draw[<->] (n0) -- (n1);
  \draw[->] (n0) -- (n2);
  \draw[->] (n1) -- (n2);
  \draw[->] (n2) -- (n3);
  \draw[->] (n3) -- (n1);
  \end{tikzpicture}
  \vspace*{-0.2cm}\captionof{figure}{Graph $G'$}
\end{center}
\end{minipage}

%

\begin{restatable}{lemma}{lemsoundnessalphaautomaton}\label{lem:soundness-alpha-automaton}
$\delta_g(\alpha(\pgraphof{msc(w)}),\dag\amessage^{p\to q})=\alpha(\pgraphof{msc(w\cdot\dag\amessage^{p\to q})})$.
%
\end{restatable}

Before we prove Lemma~\ref{lem:soundness-alpha-automaton}, we need to introduce
a few notions and observations.
Let $G=(V,E,\lambda_S,\lambda_R)$ be a P-graph.
A vertex $v\in V$ is X-covered if for all $p\in\lambda_X(v)$
$p$ is X-redundant. We also say that $v\in V$ is covered if it is
both S-covered and R-covered. A partial abstraction
of $G$ is a graph $G'=(V',E',\lambda_S,\lambda_R)$ such that
\begin{itemize}
\item $V'=\{V_1,\dots,V_n\}$ where each $V_i$ is a (not necessarily maximal)
strongly connected component of $G$, all $V_i$ are disjoints, and
for all $v\in V\setminus\bigcup_{i=1}^nV_i$ $v$ is covered.
\item for all $i,j$, $(V_i,V_j)\in E'$ iff $(v,v')\in E$ for some $v\in V_i$
and some $v'\in V_j$.
\item for all $i,X$, $\lambda_X(V_i)=\bigcup_{v\in V_i}\lambda_X(v)$
\end{itemize}
Intuitively, $G'$ is a partial abstraction of $G$ if it results from a ``partial application'' of the functions $\mergeof{.}$, $\eraseof{.}$, and
$\sweepof{.}$: some vertices of a same SCC are merged, but not necessarily all,
some labels are erased, but not neces\-sarily all, and some vertices are 
sweeped, but not necessarily all. From this observation, it follows
the following: if $G'$ is a partial abstraction of $G$, then $\alpha(G')=\alpha(G)$.
\begin{proof}
Let $w\in\Sigma^*$ and $\dag\Message{\amessage}{p}{q}$ be fixed.
Let $G_1=\pgraphof{msc(w)}$ and $G_2=\pgraphof{msc(w\cdot\dag\Message{\amessage}{p}{q})}$
and let us compare $G_1$ and $G_2$. First, there is an extra vertex $v_0$
in $G_2$ that represents $\dag\Message{\amessage}{p}{q}$, with
$\lambda_S(v_0)=\{p\}$ and either $\lambda_R(v_0)=\{q\}$ (if $\dag=!?$)
or $\lambda_R(v_0)=\emptyset$ (if $\dag=?$).
Now, consider the extra edges. Obviously, these extra edges have
$v_0$ either as source or as destination. First consider the edges
with $v_0$ as destination. The send event of $v_0$ happens after
all send events of $p$, so for all $v\neq v_0$ such that $p\in\lambda_S(v)$,
$(v,v_0)\in E_2$.
In the case where $\dag=!?$, the receive event of $v_0$ also happens
after all send and receive events of $q$, so for all
$v\neq v_0$ such that $q\in\lambda_S(v)\cup\lambda_R(v)$, $(v,v_0)\in E_2$.
There are no
other incoming edges in $v_0$
Now, consider the outgoing edges of $v_0$. The send event of $v_0$ happens
before all receive events of $p$, so for all
$v\neq v_0$ such that $p\in\lambda_R(v)$, $(v_0,v)\in E_2$.
To sum up, we have:

$$
\begin{array}{lll}
E_2&=&E_1\cup\{(v,v_0)\mid p\in\lambda_S(v)\}\\
&\cup&\{(v_0,v)\mid p\in\lambda_R(v)\}\\
&\cup&\left\{\begin{array}{ll}
\{(v,v_0)\mid q\in\lambda_S(v)\cup\lambda_R(v)\} & \mbox{if }\dag=!?\\
\emptyset & \mbox{if }\dag=!
\end{array}\right.
\end{array}
$$

Observe now that the rules to add vertices and edges to go from $G_1$
to $G_2$ are exactly the same as the rules to go from $G$ to $G'$
in the definition of $\delta_g(G,\dag\Message{\amessage}{p}{q})$.
Assume that $G=\alpha(G_1)=\alpha(\pgraphof{msc(w)}$.
Then $G'$ is a partial abstraction of $G_2$. So by the discussion above,
$$\alpha(G')=\alpha(G_2).$$
Now, by definition of $\delta_g$,
$\delta_g(G,\dag\Message{\amessage}{p}{q}))=\alpha(G')$.
To sum up
$$
\delta_g(G,\dag\Message{\amessage}{p}{q})
=\alpha(G_2)=\alpha(\pgraphof{msc(w\cdot\dag\Message{\amessage}{p}{q})}).
$$
\qed
\end{proof}

\begin{restatable}{theorem}{thmprimereg}\label{thm:prime-reg}
There is an effective deterministic finite state automaton $\mathcal A$
with less than $2^{6|\procSet|^2}$ states
such that $\alanguage(\mathcal A)=\{w\in\Sigma^*\mid msc(w)\mbox{ is prime}\}$.
\end{restatable}
\begin{proof}
Let $\anauto=(Q,\Sigma,\delta_q,q_0,F)$ be defined
by
\begin{itemize}
\item $Q=\{\alpha(\pgraphof{\msc})\mid \msc \mbox{ is an exchange}\}$;
\item $\delta_q$ as defined in Section~\ref{sec:prime};
\item $q_0=\alpha(\pgraphof{\epsilon})$ where $\epsilon$
    denotes the empty MSC
\item $F=\{G\in Q \mid |G|=1\}$
\end{itemize}
Then by Lemma~\ref{lem:finite-abstraction}, $\anauto$
is a deterministic finite state automaton with at most
$2^{6|\procSet|^2}$ states.
Moreover, by Lemma~\ref{lem:soundness-alpha-automaton},
for all $w$, $\delta_q^*(q_0,w)=\pgraphof{msc(w)}$,
so $w$ is accepted iff $\pgraphof{msc(w)}$ is
a single vertex graph. By Lemma~\ref{lem:sc-if-single-vertex},
this is equivalent to the fact that $msc(w)$ is prime.
\qed
\end{proof}

\section{Computation of $k$}\label{sec:computation-k0}

So far we have shown:  in Lemma~\ref{lemma-lpfe}, we
established that a way to compute $\sd{\system}$ was
to compute the length $k$ of
the largest prime reachable exchange. To every word $w\in\Sigma^*$,
we associated an MSC $msc(w)$, and we showed that for every reachable
MSC $\msc$, there exists a word $w\in\Sigma^*$ such that $\msc=msc(w)$
(Lemma~\ref{lem:surjectivite-encodage}). We deduced that $k$
corresponds to the length of the longest word of
$\reachlanguage\cap\primelanguage$
, if $\reachlanguage\cap\primelanguage$ is finite,
otherwise $k=\infty$. In Section~\ref{sec:rfe}, we showed
that $\reachlanguage$ is an effective regular language, and,
in Section~\ref{sec:prime}, we showed that $\primelanguage$ is
also an effective regular language. We deduce that
$\reachlanguage\cap\primelanguage$ is therefore an effective regular language,
and that $k$ is computable (since the finiteness and
the length of the longest word of a regular language are computable).
With a  careful analysis of the automata that come into
play, we can give an upper bound on $k$.

\begin{restatable}{theorem}{thmcomputationk}\label{thm:computation-k0}
$\sd{\system}$ is computable, and if $\sd{\system}<\infty$ then
$\sd{\system}<|\system|^22^{8|\procSet|^2}$, where
$|\system|$ is the number of
global control states
and $|\procSet|$ the number of processes.
\end{restatable}
\begin{proof}
The fact that $\sd{\system}$ is computable
is explained at the beginning of
Section~\ref{sec:computation-k0}. We therefore
only prove the claim that, when $k<\infty$
it holds that $k<|\system|^22^{8|\procSet|^2}$.
$k$ is the length of the longest word
in $\reachlanguage\cap\primelanguage$.
By Theorem~\ref{th:reachable},
$$
\reachlanguage =  \bigcup\mathcal{R}
$$
and by Theorem~\ref{thm:prime-reg} there is an automaton
$\anauto$
such that
$\Language{\anauto}=\primelanguage$.
So we need to bound the length of the longest
word of
$$
    \mathcal{L}(\Aut{SR}(\globalstate{l}, \globalstatelj, \globalstate{l'}))
    \cap \Language{\anauto(B,B')}
    \cap \Language{\anauto}
$$
assuming that this language is finite, for any
$\globalstate{l},\globalstate{mid},\globalstate{l'},B,B'$.
This bound is given by the number of states of any
automaton that accepts this language (since any longer word
would require the automaton to feature a loop, and the language
would not be finite).
This language is recognized by an automaton that is the product
of the automata $\Aut{SR}(\globalstate{l},\globalstate{mid},\globalstate{fin})$,
$\anauto(B,B')$, and $\anauto$, so its number of states
is bounded by
$$
    |\Aut{SR}(\globalstate{l},\globalstate{mid},\globalstate{l'})|
    \times |\anauto(B,B'))|
    \times |\anauto|
$$
By definition of $\Aut{SR}$, $L_{\Aut{SR}}=L_{\system}^2$,
so $|\Aut{SR}|\leq |L_\system|^2$ (which we can also write
$|\system|^2$).
By definition of $\anauto(B,B')$,
$L_{(B,B')}=\absBrSet=(2^{\procSet}\times2^{\procSet})^{\procSet}$,
so $|\anauto(B,B')|\leq 2^{2|\procSet|^2}$.
Finally, by Theorem~\ref{thm:prime-reg},
$|\anauto|\leq 2^{6|\procSet|^2}$.
All toghether,
$$
\begin{array}{ll}
k & \leq |\system|^22^{2|\procSet|^2}2^{6|\procSet|^2}
\\ & \leq |\system|^22^{8|\procSet|^2}
\end{array}
$$
\qed
\end{proof}

As an immediate consequence of Theorems~\ref{thm:ksync-for-fixed-k}
and \ref{thm:computation-k0}, we get the following.

\begin{restatable}{theorem}{thmmain}\label{thm:decidability-exists-k}
The following problem is decidable : given a system $\system$,
does there exists a $k$ such that $\system$ is $k$-synchronizable.
\end{restatable}






%
%
%

\section{Conclusion}\label{sec:conc}
We established that it is possible to determine whether there exists a bound
$k$ such that a given communicating system is \kSable{k}. For this, we showed
how the set of sequences of actions that compose an exchange of arbitrary size
can be represented as a regular language, which was possible thanks to the
mailbox semantics of communications. We leave for future work to decide whether it would be possible
to extend our result to peer-to-peer semantics.

\newpage
\bibliographystyle{splncs04}
 \bibliography{biblio}
 \newpage

%

\end{document}